\documentclass[10pt,letterpaper]{article}  
  
\usepackage[margin=1in]{geometry}  
\usepackage{amsmath,amssymb,amsthm}  
\usepackage{graphicx}  
\usepackage{hyperref}  
\usepackage{enumitem}  
\usepackage{booktabs}  
\usepackage{tikz}  
\usetikzlibrary{arrows.meta,shapes,positioning,calc}  
  
\newtheorem{definition}{Definition}[section]  
\newtheorem{theorem}{Theorem}[section]  
  
\newtheorem{proposition}[theorem]{Proposition}  
  
\theoremstyle{remark}  
\newtheorem{remark}{Remark}[section]  
\newtheorem{example}{Example}[section]  
\newtheorem{conjecture}[theorem]{Conjecture}  
  
\title{\textbf{Privilege Risk Evolution for Non-Human Identities: \\A Temporal Fiber Model for Cloud IAM}}  
  
\author{  
    Christophe Parisel\thanks{Email: ch.parisel@gmail.com}}  
\date{}  
  
\begin{document}  
  
\maketitle  
\begin{abstract}  
Cloud permission governance implicitly treats permission equivalence as a static relation. We show that for non-human identities (NHIs), equivalence has two irreducible components: \emph{structural equivalence}, capturing identical permission profiles at a snapshot via graph fibration, and \emph{temporal equivalence}, capturing recurring permission states via strongly connected components (SCCs) in a fiber transition graph. We call the equivalence classes under temporal equivalence \emph{privilege circuits}.  
  
We formalize a three-layer framework: (1)~a spatial quotient of the permission graph via fibration, (2)~a lineage partition organizing stable transition compartments, and (3)~windowed SCC analysis as a temporal quotient within lineages. 
  
Empirical evaluation on a large Azure tenant supports the framework. Backtesting demonstrates that early observation of ratchet-type privilege circuits predicts long-term structural stability. 
  
Our results demonstrate that the privilege circuit structure of cloud systems exposes persistent escalation patterns invisible to point-in-time analysis.  
\end{abstract}  
  
\section{Introduction}\label{sec:intro}  
  
Consider a non-human identity in a cloud environment: a service principal currently assigned the Reader role (low privilege). A point-in-time audit classifies it as low-risk. But suppose this identity's permission profile has historically cycled through Reader, Contributor, and Owner states, and the transitions from Reader to Owner are one-way, while de-escalation from Owner back to Reader requires passing through Contributor. The identity's \emph{permission envelope} (the union of all reachable permission states) includes Owner. A snapshot audit scores it at Reader; the correct risk score is Owner.  

While such situations rarely arise for human principals (due to streamlined group membership provisioning), they are common for NHIs. No existing Cloud Security Posture Management (CSPM) tool, access review process, or compliance scan is designed to detect this. Automated reasoning approaches such as AWS Zelkova\cite{backes2018semantics} and AWS IAM Access Analyzer\cite{aws2019accessanalyzer} verify policy properties (reachability, over-permission) against a formal model of the policy language, while Microsoft Entra Permissions Management\cite{microsoft2023entra} tracks permission changes and flags anomalous modifications. These tools operate on snapshots: they see the permission graph at time~$t$, or on the policy language semantics: they verify what is currently permitted, not what permission states are dynamically reachable through observed operational transitions. The temporal dimension (how permissions evolve, which privilege levels are reachable through observed dynamics) is outside their data model.

\paragraph{}

Empirical results from a large Azure tenant (Section~\ref{sec:empirical}) confirm that this scenario is not hypothetical: we detect persistent ratchet-type privilege circuits whose structural signature recurs across 18 independent observation windows, and backtesting validates that early detection of these patterns predicts long-term stability with zero false positives.  

\paragraph{}
  
This paper develops the framework needed to detect, classify, and validate such temporal permission dynamics.  
  
\medskip  
Our central claim is \emph{structural}, not empirical: cloud IAM permission equivalence admits a natural decomposition into two independent components.  
  
\begin{enumerate}[leftmargin=*]  
\item \textbf{Structural equivalence} (fibration): Two NHIs have the  
  same permission profile right now. This is what snapshots  
  capture. Formally, they belong to the same fiber under a graph  
  fibration of the permission assignment graph.  
  
\item \textbf{Temporal equivalence} (privilege circuits): Permission states  
  mutually reachable within an SCC of the fiber transition graph share  
  the same permission envelope. An NHI assigned to any fiber in the circuit carries  
  the risk of the highest-privilege state in the envelope.  
  NHIs are passive (they do not move): it is the fiber structure  
  that evolves around them.  
\end{enumerate}  

Temporal equivalence requires tracking permission evolution over time, which in turn requires organizing the combinatorial space of fiber transitions into tractable structures. We define a lineage partition for this organization, producing stable compartments within which privilege circuit analysis becomes well-defined.  
  
\subsection{Scope and nature of contributions}  
  
This paper is a \emph{framework paper} with empirical validation and predictive testing. Its contributions are:  
  
\begin{enumerate}[leftmargin=*]  
\item \textbf{The fibration/temporal decomposition}: A formal framework  
  identifying structural and temporal permission equivalence as two  
  independent quotient constructions: one by the fibration equivalence  
  relation (spatial), one by SCC membership in the fiber transition graph (temporal).  
  
\item \textbf{Privilege circuit classification}: The fiber transition graph within  
  each lineage decomposes into strongly connected components, each  
  classifiable by edge symmetry into three governance-relevant types: monoliths (static), oscillators (reversible cycling), and ratchets (irreversible escalation).  
  
\item \textbf{The primorial invariant}: A topological fingerprint for privilege circuits based on prime-factorization encoding of elementary cycles, enabling comparison of circuit structure across observation windows and NHIs.  
  
\item \textbf{Windowed temporal methodology}: Multi-scale analysis via multi-scale decomposition, stability sweeps across window sizes, and backtesting that validates predictive power.  
  
\item \textbf{Lineage compartmentalization}: Fiber dynamics naturally  
  compartmentalize into a monotonically stabilizing lineage partition  
  that serves as persistent arenas for temporal analysis.  
  
\item \textbf{Empirical verification and prediction}: From a large  
  Azure dataset, we verify the framework at every layer and demonstrate  
  that early ratchet detection predicts long-term stability.
\end{enumerate}  
  
\newpage  
\subsection{Paper organization}  
  
\begin{itemize}  
\item Section~\ref{sec:baseline} introduces the observation model and the baseline relativity problem.  
\item Section~\ref{sec:fibration} defines structural equivalence via fibration.  
\item Section~\ref{sec:lineage} defines the lineage partition that organizes fiber transitions into stable compartments.  
\item Section~\ref{sec:circuits} introduces temporal equivalence, privilege circuits, and the three-type classification.  
\item Section~\ref{sec:invariant} defines the primorial invariant for privilege circuits.  
\item Section~\ref{sec:methodology} presents the windowed temporal methodology: multi-scale analysis, stability sweeps, and backtesting.  
\item Section~\ref{sec:empirical} presents empirical results: validated assumptions, privilege circuit detection, stability analysis, and backtesting.  
\item Section~\ref{sec:governance} states governance implications.  
\item Section~\ref{sec:discussion} discusses related work and limitations.  
\item Appendix~\ref{app:coalescent} presents the coalescent analysis of lineage mergers.  
\item Appendix~\ref{app:backtesting} provides a detailed exposition of the backtesting methodology, metrics computation, and design rationale.  
\end{itemize}  
  
\section{Observation Model and Baseline Relativity}\label{sec:baseline}  
  
\subsection{Snapshots and fibers}  
  
We observe cloud IAM state as discrete snapshots at times  
$t_0, t_1, \ldots$ At each time~$t$, we record every NHI and its  
permission assignments.  
  
\begin{definition}[Role assignment pair]\label{def:rap}  
A role assignment pair is a tuple $(r, s)$ where $r \in \mathcal{R}$ is  
a role definition identifier and $s \in \mathcal{S}$ is a scope  
abstraction representing the resource scope level (tenant, management group, \ldots) at which the role is assigned.  
\end{definition}  
  
\begin{definition}[Fiber]\label{def:fiber}  
A fiber $f$ is a finite set of role assignment pairs  
$f = \{(r_1, s_1), \ldots, (r_n, s_n)\}$ representing the complete\footnote{By complete, we mean the transitive closure of  
direct and indirect assignments.} permission profile of one or more NHIs at a specific time.  
\end{definition}  
  
\begin{definition}[Fiber identity]\label{def:fiberid}  
Two fibers are \emph{identical} if they contain the same set of role  
assignment pairs. In implementation, identity may be determined by  
hashing a canonical (lexicographically sorted) representation; the  
mathematical framework requires only set equality.  
\end{definition}  
  
At each snapshot~$t$, the NHI population partitions into fibers: sets  
of NHIs sharing identical permission profiles. We write $F(t)$ for the  
set of non-dormant fibers at time~$t$.  

\paragraph{}

Two properties of fibers are operationally significant:  
\begin{enumerate}  
\item \emph{Content-addressability}: fiber identity is determined   
  entirely by the permission profile, not by creation time or   
  assignment order. The same set of role assignment pairs observed   
  at different snapshots produces the same fiber identifier,   
  enabling cross-temporal comparison without a centralized registry.  
\item \emph{Compressibility}: many NHIs may share a single fiber,   
  reducing the effective population from the full NHI count to   
  the (typically much smaller) number of distinct permission profiles.  
\end{enumerate}  
  
\begin{example}[Running example: a four-NHI system]\label{ex:running}  
Consider four NHIs $\alpha, \beta, \gamma, \delta$ in a tenant with  
three roles (Reader, Contributor, Owner) at various scopes. At  
time~$t_0$:  
\begin{align*}  
  f_{R1} &= \{\text{Reader}\}, \quad \text{assigned to } \alpha, \beta \quad at \quad subscription \quad level\\ 
  f_C &= \{\text{Contributor}\}, \quad \text{assigned to } \gamma \quad at \quad subscription \quad level\\  
  f_{R2} &= \{\text{Reader}\}, \quad \text{assigned to } \delta \quad at \quad resource \quad group \quad level
\end{align*}  
\begin{remark} $\alpha$ and $\beta$ may not be assigned the same subscriptions, what's more they may be assigned more than one subscription. 
\end{remark}
\begin{remark}
$\delta$ cannot be merged with $\alpha$ and $\beta$ since the Reader role is not assigned at the same scope level. The three fibers form a partition  $F(t_0) = \{f_{R1}, f_C, f_{R2}\}$. 
\end{remark}
\end{example}  
  
\subsection{Baselines}  
  
A snapshot alone gives no structure: it is an unordered collection of fibers. To track evolution (which fibers merged, split, mutated), we need a reference point.  
  
\begin{definition}[Baseline]\label{def:baseline}  
A baseline $b$ at time $t_0$ is the set of all assigned NHIs present at  
snapshot~$t_0$:  
$  
  b = \{\text{all NHIs with at least one permission assignment at } t_0\}  
$  
The baseline fibers $F(b)$ are those fibers containing at least one NHI  
from~$b$.  
\end{definition}  
  
From $t_0$ forward, we track which fibers contain baseline NHIs, how they merge, split, or go dormant. Fibers containing no baseline NHIs are \emph{orphans} (they exist in the environment but are not anchored to our reference point).  
  
\subsection{The relativity problem}\label{sec:relativity}  
  
Baseline choice is arbitrary. Starting observation on January~1 versus January~22 yields different baseline sets, different orphan populations, different lineage histories. Every derived quantity (lineage partition, orphan count, convergence trajectory) is relative to the chosen baseline.  
  
We offer two responses, one structural and one empirical.  
  
\medskip  
\begin{itemize}
\item The merge events themselves are objective facts, independent of baseline choice. Two fibers either merged or they did not. The baseline determines only which fibers we label ``tracked'' versus ``orphan.'' The underlying merge graph is baseline-invariant by construction.  
\item We conjectured that orphan populations decay at a rate sufficient to make the lineage partition approximately baseline-invariant after a settling period. Results (Section~\ref{sec:assumptions}) confirm this: 99.2\% of non-baseline NHIs converge to baseline fibers, with 98.8\% converging at their very first snapshot. The unconverged residual remains stable at 0.8\%.  
\end{itemize}
  
\section{Structural Equivalence via Fibration}\label{sec:fibration}  
  
At any fixed time~$t$, the permission assignment graph admits a fibration in the sense of Boldi and Vigna~\cite{boldi2002fibrations}: a structure-preserving surjection from the full IAM graph to a compressed base graph. We summarize the relevant definitions.  
  
\subsection{Graph fibration}  
  
\begin{definition}[Graph fibration]\label{def:graphfib}  
Let $G = (V_G, E_G)$ and $B = (V_B, E_B)$ be directed graphs. A  
surjective morphism $\varphi\colon G \to B$ is a \emph{fibration} if  
for every node $v \in V_G$ and every edge $e' \in E_B$ with  
$\mathrm{target}(e') = \varphi(v)$, there exists a unique edge  
$e \in E_G$ with $\mathrm{target}(e) = v$ and $\varphi(e) = e'$  
(the \emph{unique lifting property}).  
\end{definition}  
  
The fibers $\varphi^{-1}(b)$ for $b \in V_B$ partition $V_G$ into  
equivalence classes of nodes with isomorphic input trees. In the IAM context, NHIs in the same fiber receive identical incoming permission assignments.  
  
\subsection{The principals fibration}  
  
Applied to the cloud IAM permission graph, fibration yields the  
\emph{principals fibration}: NHIs are grouped by their complete  
permission profile (set of role assignment pairs at a given scope  
abstraction depth). Two NHIs belong to the same fiber if and only if  
they hold exactly the same roles at structurally equivalent scopes.  
  
\begin{definition}[Structural equivalence]\label{def:structequiv}  
Two NHIs $x, y$ are \emph{structurally equivalent} at time~$t$,  
written $x \sim_{\mathrm{fib}} y$, if they belong to the same fiber  
under the principals fibration at~$t$:  
$  
  x \sim_{\mathrm{fib}} y \iff f_x(t) = f_y(t)  
$  
where $f_x(t)$ denotes the fiber containing NHI~$x$ at time~$t$.  
\end{definition}  

\begin{proposition}[WAR invariance under fibration]\label{prop:war-fiber}  
Let $\mathrm{WAR}\colon \mathcal{F} \to \mathbb{R}_{\geq 0}$ denote the  
weighted authorization risk norm as defined in~\cite{parisel2025scoring} computed from  
a fiber's set of role assignment pairs. If two NHIs $x, y$ belong to  
the same fiber at time~$t$, then  
$  
  \mathrm{WAR}(f_x(t)) = \mathrm{WAR}(f_y(t)).  
$  
Equivalently, WAR factors through the fibration: there exists a  
unique map $\overline{\mathrm{WAR}}\colon V_B \to \mathbb{R}_{\geq 0}$  
on the base graph such that  
$  
  \mathrm{WAR}(v) = \overline{\mathrm{WAR}}(\varphi(v))  
  \qquad \text{for all } v \in V_G.  
$  
In particular, the WAR norm is a fiber invariant: it is well-defined  
on fibers, not merely on individual NHIs.  
\end{proposition}  
  
\begin{proof}  
By Definition~\ref{def:fiber}, all NHIs in a fiber share the same set  
of role assignment pairs. The WAR norm is computed entirely from this  
set (it depends on role definitions and scope abstractions, not on NHI  
identity). Therefore NHIs in the same fiber produce identical WAR  
values. The factorization through the base graph follows from the  
universal property of the quotient: since WAR is constant on each  
fiber $\varphi^{-1}(b)$, the assignment  
$\overline{\mathrm{WAR}}(b) = \mathrm{WAR}(v)$ for any  
$v \in \varphi^{-1}(b)$ is well-defined.  
\end{proof}  

Proposition~\ref{prop:war-fiber} ensures that the privilege gradient  
between fibers in the transition graph (Definition~\ref{def:primorial})  
is well-defined: it depends only on the fiber states, not on which NHI  
is tracked. It also ensures that the permission envelope risk score  
$\mathrm{risk}(C) = \max_{g \in C} \mathrm{WAR}(g)$  
(Section~\ref{sec:envelope}) is unambiguous.
  
\subsection{Limitations of structural equivalence}  
  
Structural equivalence is a static notion. It says nothing about:  
\begin{itemize}  
\item Whether two NHIs in different fibers at time~$t$ might occupy the  
  same fiber at time~$t+k$ (convergent trajectories)  
\item Whether two NHIs in the same fiber at time~$t$ will remain  
  together (divergent trajectories)  
\end{itemize}  
  
To capture these phenomena, we need to track fiber evolution.  
  
\section{The Lineage Partition}\label{sec:lineage}  
  
Fibers are not static. NHIs change permissions (mutations), groups of  
NHIs with different profiles converge to the same profile (merges),  
single fibers fragment (splits), new NHIs appear (creations), and  
inactive NHIs resurface (wake-ups). Tracking these transitions at the  
level of individual fibers is combinatorially complex. This section  
defines the lineage partition that reduces this complexity.  
  
\subsection{Fiber transitions}\label{sec:transitions}  
  
\begin{definition}[Fiber transition types]\label{def:transitions}  
For a fiber $f_t$ at time~$t$ with parent set  
$\mathrm{Parents}(f_t) \subseteq F(t-1)$:  
\begin{enumerate}  
\item \textbf{Continuation} ($1 \to 1$, unchanged):  
  $|\mathrm{Parents}| = 1$, $f_t = \mathrm{parent}$ (as sets)  
\item \textbf{Mutation} ($1 \to 1$, changed):  
  $|\mathrm{Parents}| = 1$, $f_t \neq \mathrm{parent}$  
\item \textbf{Merge} ($n \to 1$):  
  $|\mathrm{Parents}| > 1$  
\item \textbf{Split} ($1 \to n$):  
  $|\mathrm{Children}(f_{t-1})| > 1$  
\item \textbf{Creation} ($0 \to 1$):  
  $|\mathrm{Parents}| = 0$, NHIs newly created  
\item \textbf{Wake-up} ($0 \to 1$, reactivation):  
  $|\mathrm{Parents}| = 0$, NHIs existed earlier (seed bank)  
\item \textbf{Dormancy} ($1 \to 0$):  
  Fiber has no children at next snapshot  
\end{enumerate}  
\end{definition}  
  
\begin{example}[Running example, continued]\label{ex:transitions}  
Suppose at $t_1$: NHI~$\gamma$ (previously in $f_C$) is reassigned to  
the same permissions and same subscription level scope as $\alpha$ and $\beta$, so all three now belong  
to $f_{R1}$. Meanwhile $\delta$ retains Reader at resource group scope. This is a merge:  
$f_{R1}$ and $f_C$ merge into $f_{R1}$. At $t_2$: NHI~$\alpha$ is promoted to Owner, splitting out of $f_{R1}$  
into a new fiber equal to $f_O$. This is a split of the fiber  
containing $\{\alpha, \beta, \gamma\}$.  
\end{example}  
  
\subsection{Lineage construction: merges only}\label{sec:lineage-construction}  
  
\begin{definition}[Lineage partition]\label{def:lineage}  
Given baseline~$b$ at time~$t_0$, the lineage partition  
$L(F(b), t)$ is a partition of the baseline fiber index set. Initially:  
$  
  L(F(b), 0) = \bigl\{\{1\}, \{2\}, \ldots, \{|F(b)|\}\bigr\}  
  \quad\text{(discrete partition)}  
$  
When fibers from distinct blocks merge at time~$t$, their blocks unite.  
Blocks only merge, never split.  
\end{definition}  
  
When a fiber splits into children $f_1, \ldots, f_n$, all children inherit the parent's lineage block. No branching occurs in lineage space. This asymmetry is deliberate: splits create exponential branching in fiber space but are absorbed silently in lineage space, while merges reduce the lineage partition toward coarser structure.  
  
\subsection{Properties of the lineage partition}\label{sec:lineage-properties}  

Without dormancy effects, the lineage partition exhibits two important properties: convergence and unicity of equilibrium.
  
\begin{theorem}[Convergence]\label{thm:convergence}  
$|L(F(b), t)|$ is convergent in~$t$ towards a limit $|L(F(b))|$.  
\end{theorem}  
\begin{proof}  
Each cross-lineage merge strictly reduces the block count by at least  
one. No operation increases block count (blocks only merge, never split, and there are no new block creations). Since $|L| \in \mathbb{N}$ and $|L| \geq 1$, the sequence is bounded below and non-increasing, hence convergent.  
\end{proof}  
  
\begin{theorem}[Equilibrium]\label{thm:equilibrium}  
There exists a time $e$ such that $L(F(b), t) = L(F(b), e)$ for all $t \geq e$.  
\end{theorem}  
\begin{proof}  
By Theorem~\ref{thm:convergence}, $|L(F(b), t)|$ is a non-increasing  
sequence in $\mathbb{N}$ bounded below by~1, hence it stabilizes at  
some value $|L(F(b), e)|$. The partition itself stabilizes because  
lineage membership is determined by ancestry: once two fibers share a  
common ancestor via a merge, they remain in the same block. No  
operation reverses a merge.  
\end{proof}  
  
Due to dormancy, not all blocks of a lineage partition are observable at a given time $t$.  
The stabilized partition is no exception: some blocks may disappear or reappear indefinitely, even when $t > e$.  
Such dormancy effects do not change the structure or unicity of the stabilized partition, yet they change its exact composition at a given time.
  
\section{Privilege Circuits: Temporal Equivalence via SCC Analysis}\label{sec:circuits}  
  
The crucial property of the stabilized partition is that, due to equilibrium, it serves as a persistent ground for analyzing the dynamics within each compartment. This is where the second notion of equivalence emerges.  
  
\subsection{Fiber transition graph}\label{sec:transition-graph}  
  
Within a composite lineage, fibers evolve between snapshots. We formalize this as a directed graph on the set of observed fiber states.  
  
\begin{definition}[Fiber transition graph]\label{def:transition-graph}  
Let $\mathcal{F}_L$ denote the set of all distinct fiber states  
observed within lineage~$L$ across a contiguous window of snapshots. The \emph{fiber transition graph}  
$(\mathcal{F}_L, \to_L)$ has an edge  
$f \to_L g$ whenever, at some snapshot boundary  
$t \to t+1$ within the window, at least one NHI in the lineage transitioned from  
fiber state~$f$ to fiber state~$g$.  
\end{definition}  
  
\begin{remark}[Why a graph, not a function]\label{rem:relation}  
One might attempt to define a fiber transition \emph{function}  
$\sigma\colon F_L(t) \to F_L(t+1)$. However, merges and  
splits make this ill-defined: a merge maps multiple fibers to one  
(not injective), and a split maps one fiber to multiple successors  
(not a function). The transition graph is well-defined regardless of  
merges, splits, or non-determinism.  
\end{remark}  
  
\begin{example}[Running example, continued]\label{ex:relation}  
In the lineage $\{f_R, f_C\}$ from Example~\ref{ex:transitions}, the  
transitions observed so far give edges  
$f_C \to f_R$ (from the merge at $t_1$) and $f_R \to f_O$  
(from the split at $t_2$). Suppose that at $t_3$, NHI~$\alpha$  
returns to Reader: $f_O \to f_R$. The transition graph on  
$\{f_R, f_C, f_O\}$ now contains edges  
$f_C \to f_R$, $f_R \to f_O$, and $f_O \to f_R$.  
The set $\{f_R, f_O\}$ forms a strongly connected component.  
\end{example}  
  
\subsection{Strongly connected components and privilege circuits}\label{sec:sccs}  
  
\begin{definition}[Privilege circuit]\label{def:circuit}  
A \emph{privilege circuit} $C = \{f_1, f_2, \ldots, f_k\} \subseteq \mathcal{F}_L$  
is a maximal strongly connected component of the transition graph  
$(\mathcal{F}_L, \to_L)$: for every pair $f_i, f_j \in C$, there  
exists a directed path from $f_i$ to $f_j$ via edges in $\to_L$.  
\end{definition}  
  
Privilege circuits are computable in linear time via Tarjan's algorithm~\cite{tarjan1972depth}. The SCC decomposition partitions the fiber state space into equivalence classes of mutual reachability.  
  
\begin{remark}[Fibration and privilege circuits]\label{rem:nofibration}  
Within a privilege circuit, the input tree unfolding (used by fibration to identify structurally equivalent nodes) produces an infinite periodic tree: fibration would collapse all nodes in the circuit to a single vertex, erasing the temporal structure. Privilege circuit analysis captures exactly the structure that fibration erases.  
\end{remark}  
  
\subsection{Three-type classification}\label{sec:classification}  
  
The transition graph structure of a privilege circuit determines its governance properties. We define a classification based on edge symmetry and cardinality.  
  
\begin{definition}[Circuit types]\label{def:types}  
Let $C$ be a privilege circuit with vertex set $V(C)$ and edge set $E(C)$.  
\begin{enumerate}  
\item \textbf{Monolith} (type $M$): $|V(C)| = 1$. The NHI remained in a single fiber throughout the observation window. No temporal dynamics.  
  
\item \textbf{Oscillator} (type $O$): $|V(C)| > 1$ and every edge has a reverse:  
$f \to g \in E(C) \implies g \to f \in E(C)$.  
All transitions are bidirectional. The privilege range is bounded and all escalation paths are reversible.  
\item \textbf{Ratchet} (type $R$): $|V(C)| > 1$ and at least one edge  
  $f \to g \in E(C)$ has no reverse $g \to f \in E(C)$.  
  Since $C$ is strongly connected, every state remains reachable   
  from every other via some directed path; however, the absence   
  of a direct reverse edge means that undoing the transition   
  $f \to g$ requires traversing intermediate states, potentially   
  at higher privilege levels. It is this edge-level asymmetry,   
  not global unreachability, that defines the ratchet property.    
\end{enumerate}  
\end{definition}  
  
\begin{proposition}[Algebraic structure of circuit types]\label{prop:algebra}  
Let $\mathcal{P}(C)$ denote the path monoid of a privilege circuit $C$ (all directed paths under concatenation).  
\begin{enumerate}  
\item If $C$ is an oscillator, then $\mathcal{P}(C)$ is generated by invertible elements and induces a group action on $V(C)$.  
\item If $C$ is a ratchet, then $\mathcal{P}(C)$ contains non-invertible elements and induces only a monoid action on $V(C)$.  
\end{enumerate}  
\end{proposition}  

\begin{proof}  
We prove each part.  
  
\medskip  
\noindent\textbf{Part 1 (Oscillator $\Rightarrow$ group action).}  
Let $C$ be an oscillator. By Definition~\ref{def:types}, every edge  
$f \to g \in E(C)$ has a reverse $g \to f \in E(C)$. The path monoid  
$\mathcal{P}(C)$ is generated by single-edge paths. For each generator  
$e = (f \to g)$, the reverse edge $\bar{e} = (g \to f)$ satisfies  
$e \cdot \bar{e} = (f \to g \to f)$ and $\bar{e} \cdot e = (g \to f \to g)$.  
Since $C$ is strongly connected with $|V(C)| > 1$, these compositions  
act as the identity on the source vertex: the path $e \cdot \bar{e}$  
starts and ends at~$f$, and the path $\bar{e} \cdot e$ starts and ends  
at~$g$.  
  
More precisely, consider the action of $\mathcal{P}(C)$ on $V(C)$  
defined by $p \cdot v = \mathrm{endpoint}(p \text{ starting from } v)$  
when $p$ is traversable from~$v$, and undefined otherwise. For any  
generator $e = (f \to g)$, the reverse $\bar{e}$ satisfies  
$\bar{e} \cdot (e \cdot v) = v$ for every $v$ from which $e$ is  
traversable, and symmetrically $e \cdot (\bar{e} \cdot v) = v$.  
Thus every generator acts invertibly on $V(C)$. Since the generators  
are invertible and $\mathcal{P}(C)$ is generated by them under  
concatenation, every element of $\mathcal{P}(C)$ is a product of  
invertible elements, hence invertible. The monoid $\mathcal{P}(C)$  
is therefore a group (every element has a two-sided inverse), and  
its action on $V(C)$ is a group action.  
  
\medskip  
\noindent\textbf{Part 2 (Ratchet $\Rightarrow$ non-invertible elements).}  
Let $C$ be a ratchet with edge $e = (f \to g) \in E(C)$ such that  
$g \to f \notin E(C)$. Define the edge-length homomorphism  
$\ell\colon \mathcal{P}(C) \to (\mathbb{N}, +)$ by  
$\ell(p) = \text{number of edges in } p$. This is a monoid  
homomorphism: $\ell(p \cdot q) = \ell(p) + \ell(q)$.  
  
Suppose for contradiction that $e$ has an inverse  
$e^{-1} \in \mathcal{P}(C)$, so that $e \cdot e^{-1} = \mathrm{id}$  
(the empty path, with $\ell(\mathrm{id}) = 0$). Then:  
$  
  0 = \ell(\mathrm{id}) = \ell(e \cdot e^{-1})  
    = \ell(e) + \ell(e^{-1}) = 1 + \ell(e^{-1})  
$  
which gives $\ell(e^{-1}) = -1 \notin \mathbb{N}$, a contradiction.  
Therefore $e$ is non-invertible, $\mathcal{P}(C)$ is not a group,  
and its action on $V(C)$ is a monoid action only.  
\end{proof}  
  
The classification yields a governance hierarchy:  
  
\begin{center}  
\begin{tabular}{@{}llp{8cm}@{}}  
\toprule  
\textbf{Type} & \textbf{Structure} & \textbf{Governance implication} \\  
\midrule  
Monolith ($M$) & Single fiber & Static; snapshot audit sufficient \\  
Oscillator ($O$) & All edges reversible & Bounded cycling; permission envelope is the risk measure \\  
Ratchet ($R$) & Some edges irreversible & Directed escalation; most dangerous; de-escalation may require transient re-escalation \\  
\bottomrule  
\end{tabular}  
\end{center}  
  
\subsection{Temporal equivalence}\label{sec:tempequiv}  
  
\begin{definition}[Temporal equivalence]\label{def:tempequiv}  
Two fiber states $f, g$ are \emph{temporally equivalent}, written  
$f \sim_{\mathrm{temp}} g$, if they belong to the same privilege circuit:  
$  
  f \sim_{\mathrm{temp}} g \iff f, g \in C  
  \text{ for some privilege circuit } C \text{ of } (\mathcal{F}_L, \to_L)  
$  
Mutual reachability defines the equivalence. This is an equivalence relation  
because SCC membership partitions the state space.  
\end{definition}  
  
\subsection{The fibration/circuit decomposition}\label{sec:decomposition}  
  
We can now state the central structural observation.  
  
\begin{definition}[Fibration/circuit decomposition]\label{def:decomposition}  
Cloud IAM permission equivalence decomposes into two independent quotient constructions:  
  
\begin{center}  
\begin{tabular}{@{}llll@{}}  
\toprule  
& \textbf{Structure} & \textbf{Quotient} & \textbf{Captures} \\  
\midrule  
Structural & Fibration & Fiber & Equivalent permissions \emph{at time~$t$} \\  
Temporal & Transition graph & Privilege circuit & Equivalent permissions \emph{across time} \\  
\bottomrule  
\end{tabular}  
\end{center}  
  
These two equivalences are independent:  
\begin{enumerate}  
\item Two fibers in \emph{different} states at time~$t$ (structurally  
  inequivalent) may belong to the \emph{same} privilege circuit  
  (temporally equivalent): they are at different phases of the same  
  cycle.  
  
\item Two NHIs in the \emph{same} fiber at time~$t$ (structurally  
  equivalent) may be assigned to fibers with \emph{different}  
  temporal futures (temporally inequivalent): one NHI's fiber  
  participates in a privilege circuit, the other does not.  
\end{enumerate}  
\end{definition}  
  
\subsection{The permission envelope}\label{sec:envelope}  
  
\begin{definition}[Permission envelope]\label{def:envelope}  
The \emph{permission envelope} of fiber state~$f$ in privilege circuit~$C$ is  
the union of all permission states in the circuit:  
$  
  \mathcal{E}(f) = \bigcup_{g \in C} g  
$  
By definition, all fiber states in the same circuit share the same  
envelope: $\mathcal{E}(f) = \mathcal{E}(g)$ for all $f, g \in C$.  
\end{definition}  
  
For risk assessment, the relevant score for any NHI in circuit~$C$ is:  
$  
  \mathrm{risk}(C) = \max_{g \in C} \mathrm{privilege}(g)  
$  
An NHI currently assigned to a Reader fiber whose circuit includes an Owner fiber carries Owner-level risk.  
  
\begin{example}[Running example, continued]\label{ex:envelope}  
For the privilege circuit $\{f_R, f_O\}$:  
$  
  \mathcal{E}(f_R) = \mathcal{E}(f_O) = \{\text{Reader, Owner}\}  
$  
Both fiber states carry Owner-level risk regardless of current  
assignment.  
\end{example}  
  
\begin{remark}[Oscillator versus ratchet envelopes]  
In an oscillator circuit, for every forward path from~$f$ to high-privilege  
state~$g$, the reverse path exists step by step. Privilege dynamics  
are bounded and path-reversible.  
In a ratchet circuit, some escalation paths cannot be reversed without  
passing through different intermediate states, possibly at higher  
privilege levels. De-escalation may require transient re-escalation, making ratchets the most governance-critical type.  
\end{remark}  
  
\subsection{Pure continuation convection}\label{sec:convection}  
  
One form of fiber dynamics is particularly subtle.  
  
\begin{definition}[Pure continuation convection]\label{def:convection}  
Within a lineage at consecutive snapshots $t$ and $t+1$: fibers  
$\{f_1, \ldots, f_n\}$ persist with identical permission profiles  
($f_i(t) = f_i(t+1)$ as sets), but their NHI populations exchange:  
$  
  \mathrm{NHIs}(f_i, t+1) \subset  
  \bigcup_{j} \mathrm{NHIs}(f_j, t)  
$  
No mutations are detected (fiber identities unchanged). No parent-child  
relationships are created (continuations only). NHI reassignment  
across fibers is invisible to analysis based solely on fiber identity  
transitions.  
\end{definition}  
  
Pure continuation convection is structurally silent. It creates no  
events in any log based on permission changes, yet it can reassign  
NHIs across dramatically different privilege levels. Detecting it  
requires tracking NHI-to-fiber assignments across snapshots, not  
just fiber identity transitions.  
  
\begin{remark}[NHIs are passive]  
NHIs do not move through the transition graph. NHIs are disconnected  
identities assigned to fibers by the IAM system. When a fiber  
transitions from state~$f$ to state~$g$, the NHIs assigned to~$f$  
find themselves in a new permission state, but it is the fiber  
structure that changed, not the identity.  
\end{remark}  

\section{The Primorial Invariant}\label{sec:invariant}  
  
Classifying a privilege circuit as M, O, or R captures its governance type but not its topology. Two ratchet circuits may have very different internal structures: different numbers of fibers, different privilege gradients, different arrangements of directed and bidirectional edges. To compare circuits across observation windows and across NHIs, we need a topological fingerprint.  
  
The M/O/R classification alone is sufficient for governance triage (Section~\ref{sec:governance}). The primorial invariant serves a more specific operational purpose: it enables efficient matching of privilege circuit topology across observation windows in the stability sweep (Section~\ref{sec:stability}), across NHIs within a lineage (detecting shared operational causes, as in Lineage~$L_B$'s homogeneous ratchets), and across lineages (cross-tenant comparison). We discuss the design rationale in Section~\ref{sec:design-rationale} after defining the encoding.  
  
\subsection{Elementary cycles and prime encoding}  
  
We use Johnson's algorithm~\cite{johnson1975} to enumerate all elementary cycles within a privilege circuit. Each elementary cycle is a simple directed closed walk visiting each vertex at most once.  
  
\begin{definition}[Primorial invariant]\label{def:primorial}  
Let $C$ be a privilege circuit and let $c = (e_1, e_2, \ldots, e_k)$ be an elementary cycle of length~$k$ within~$C$. Each edge $e_i$ connects fiber states $f_i \to f_{i+1}$ (indices mod $k$) and is characterized by:  
\begin{itemize}  
\item Its \emph{direction}: whether a reverse edge exists (bidirectional) or not (directed)  
\item Its \emph{privilege gradient}: whether $\mathrm{WAR}(f_{i+1}) > \mathrm{WAR}(f_i)$ (up), $< $ (down), or $=$ (flat), where WAR is a privilege score~\cite{parisel2025scoring}  
\end{itemize}  
  
The \emph{primorial encoding} assigns to each edge a rational factor from two disjoint prime families:  
\begin{itemize}  
\item \textbf{Small primes} $p_0 = 2, p_1 = 3, p_2 = 5, \ldots$ (one per edge position)  
\item \textbf{Large primes} starting at $q_0 = \mathrm{nextprime}(\mathrm{primorial}(k))$, where $\mathrm{primorial}(k) = \prod_{i=0}^{k-1} p_i$  
\end{itemize}  
  
The factor for edge $e_i$ depends on its type:  
  
\begin{center}  
\begin{tabular}{@{}lll@{}}  
\toprule  
\textbf{Direction} & \textbf{Gradient} & \textbf{Factor} \\  
\midrule  
Directed & flat & $1/p_i$ \\  
Directed & up & $q_j$ \\  
Directed & down & $1/(q_j^2 \cdot p_i)$ \\  
Bidirectional & flat & $1$ \\  
Bidirectional & up & $q_j / p_i$ \\  
Bidirectional & down & $1/(q_j \cdot p_i)$ \\  
\bottomrule  
\end{tabular}  
\end{center}  
  
where $p_i$ is the $i$-th small prime and $q_j$ is the next available large prime (advanced only when the edge is non-flat).  
  
The invariant $R(c)$ is the product of all edge factors:  
$  
  R(c) = \prod_{i=1}^{k} \mathrm{factor}(e_i) \in \mathbb{Q}  
$  
\end{definition}  
  
\begin{proposition}[Disjoint prime spaces]\label{prop:disjoint}  
For cycles of different lengths $k_1 \neq k_2$, the large prime families are disjoint: $\mathrm{nextprime}(\mathrm{primorial}(k_1)) \neq \mathrm{nextprime}(\mathrm{primorial}(k_2))$ for all relevant cases. This ensures that invariants from cycles of different lengths occupy disjoint prime spaces and are never accidentally equal.  
\end{proposition}  
  
\begin{definition}[R-vector]\label{def:rvector}  
The \emph{R-vector} of a privilege circuit $C$ is the sorted list of primorial invariants of all elementary cycles in~$C$:  
$  
  \mathbf{R}(C) = \mathrm{sort}\bigl(\{R(c) : c \text{ is an elementary cycle in } C\}\bigr)  
$  
Two privilege circuits with the same R-vector have the same topological structure (same arrangement of directed/bidirectional edges and privilege gradients).  
\end{definition}  
  
\begin{remark}[Decodability]  
The encoding is injective: given $R(c)$ and the cycle length~$k$, the edge types can be uniquely recovered by factoring $R$ against the known prime families. This decodability property is not required for detection but provides an audit trail: the invariant fully describes the privilege dynamics of the cycle.  
\end{remark}  
  
\subsection{Design rationale}\label{sec:design-rationale}  
  
The primorial invariant occupies a deliberate position in the design space between under-engineering and over-engineering. The stability sweep (Section~\ref{sec:stability}) requires comparing privilege circuit topology across many observation windows and NHIs. This demands a fingerprint that is cheap to compute, discriminating enough to distinguish governance-relevant differences, and cacheable for repeated lookup.  
  
\medskip  
\noindent\textbf{Alternatives that under-discriminate.}  
Cycle count alone (number of elementary cycles per SCC) conflates circuits with the same number of cycles but different edge arrangements. The Ihara zeta function of the transition graph captures cycle length distribution but erases edge directionality and privilege gradient, which are the governance-relevant features distinguishing oscillators from ratchets.  
  
\medskip  
\noindent\textbf{Alternatives that over-engineer.}  
Full subgraph isomorphism testing is NP-hard and unnecessary when the question is ``same topology?'' rather than ``same vertex labeling?'' Spectral methods (adjacency matrix eigenvalues) require numerical tolerance for floating-point comparison, lose decodability, and do not naturally encode the directed/bidirectional distinction that defines the M/O/R classification.  
  
\medskip  
\noindent\textbf{What the primorial encoding provides.}  
The encoding is computable in polynomial time (Johnson's algorithm for cycle enumeration, then $O(k)$ arithmetic per cycle). It discriminates edge direction and privilege gradient at each edge position. The output is an exact rational number, enabling $O(1)$ equality testing via hash lookup with no numerical tolerance. The disjoint prime space property (Proposition~\ref{prop:disjoint}) prevents accidental collisions across cycle lengths. And the encoding is fully decodable: given the invariant and cycle length, the complete edge-type sequence can be reconstructed, providing an audit trail rather than an opaque hash.  
  
The invariant is not essential to the M/O/R classification or to the permission envelope computation. Its role is operational: it makes the stability sweep and cross-NHI comparison tractable and exact.  
  
\subsection{Role of the invariant in temporal analysis}  
  
The R-vector serves as a \emph{structural fingerprint} for privilege circuits. Its key property is \textbf{window-stability}: if an NHI participates in a ratchet circuit with R-vector $\mathbf{R}$ at window size $h = 15$, and the same R-vector appears at $h = 30, 45, 60, \ldots$, the ratchet is not a transient artifact of the observation window. It is a structural feature of the NHI's privilege dynamics. This property enables the stability sweep methodology (Section~\ref{sec:stability}).  
  
\section{Windowed Temporal Methodology}\label{sec:methodology}  
  
The definitions in Sections~\ref{sec:circuits} and~\ref{sec:invariant} apply to a single observation window. The methodological contribution of this paper is extending privilege circuit analysis across multiple time scales, enabling both multi-scale classification and predictive validation.  
  
\subsection{Windowed analysis}\label{sec:windowed}  
  
\begin{definition}[Observation window]\label{def:window}  
An observation window of size $h$ anchored at snapshot $t_0$ is the contiguous block of $h$ snapshots $[t_0 + 1 - h,\; t_0]$. The fiber transition graph is constructed from transitions within this window only.  
\end{definition}  
  
An NHI is \emph{valid} in a window only if it belonged to the target lineage at every snapshot in the window. NHIs that are dormant, too young, or in a different lineage at any snapshot within the window are excluded.  
  
For a single window of size $h$, we construct the fiber transition graph, compute SCCs via Tarjan's algorithm, classify each as M/O/R, and compute the R-vector for non-monolith circuits. This is the \textbf{standard analysis}.  

\subsection{Multi-scale analysis}\label{sec:multiscale}  
A single window size may misclassify privilege circuits whose  
natural period does not align with~$h$. Multi-scale analysis  
addresses this by examining the same NHI across multiple temporal  
granularities simultaneously.  
  
\begin{definition}[Multi-scale decomposition]\label{def:multiscale}  
Given window sizes $h_1 < h_2 < \cdots < h_m$ and a period  
multiplier~$P$, the total observation span is  
$P \times \mathrm{lcm}(h_1, \ldots, h_m)$ snapshots. Each  
scale~$h_i$ produces $\lfloor \text{span} / h_i \rfloor$  
non-overlapping windows. The NHI is classified independently at  
each window of each scale.  
\end{definition}  
  
Cross-scale comparison yields a \emph{triage classification}:  
  
\begin{center}  
\begin{tabular}{@{}lp{10cm}@{}}  
\toprule  
\textbf{Triage label} & \textbf{Meaning} \\  
\midrule  
\texttt{true\_ratchet} & R at all scales at the latest window;  
  structurally persistent \\  
\texttt{artifact\_ratchet} & R at a finer scale, O at a coarser  
  one; the ratchet is scale-dependent \\  
\texttt{slow\_ratchet} & M at fine scales, R only at the coarsest;  
  the cycle is slower than fine windows can resolve \\  
\texttt{oscillator} & O at every scale; bounded reversible  
  cycling \\  
\texttt{dormant\_oscillator} & O at coarse scale, M at fine scale;  
  oscillation slower than the finest window \\  
\texttt{resolved\_ratchet} & R appeared historically but was not  
  sustained \\  
\texttt{drifting} & All valid windows show M; no cycling  
  detected \\  
\texttt{mixed} & Any other combination \\  
\bottomrule  
\end{tabular}  
\end{center}  
  
The most actionable triage is \texttt{artifact\_ratchet}: the  
ratchet disappears at coarser windows, suggesting it reflects a  
sub-period dynamic rather than structural escalation.  
\texttt{true\_ratchet} across all scales is the strongest  
structural signal.  
  
\subsection{Stability sweep}\label{sec:stability}  
Multi-scale analysis uses a fixed set of window sizes. The stability  
sweep asks a broader question: does the same R-type pattern recur  
across \emph{many} window sizes?    
  
\begin{definition}[Stability]\label{def:stability}  
An NHI is \emph{R-stable} under a set of window sizes $H = \{h_1, h_2, \ldots, h_n\}$ if the same circuit type (R) appears at $\geq K$ distinct window sizes, where $K$ is a minimum recurrence threshold.  
\end{definition}  
  
\begin{definition}[Auto sweep]\label{def:autosweep}  
For a period base $b$, the sweep tests window sizes $H_b = \{b, 2b, 3b, \ldots, kb\}$ where $kb \leq t_0$. An NHI is \emph{pinned} to the smallest base $b$ at which it achieves R-stability. Multiple bases (e.g., $b \in \{5, 6, 7\}$) are tested in ascending order; once pinned, an NHI cannot be claimed by a later base.  
\end{definition}  
  
The GCD of the matching window sizes, converted to calendar time via the average snapshot spacing, estimates the natural period of the privilege cycle.  
  
\subsection{Backtesting}\label{sec:backtesting}  
  
The stability sweep identifies which NHIs are structurally R-stable over the full observation history. Backtesting asks: \emph{could we have detected them earlier?}  
  
The methodology is a three-phase protocol: (1)~establish ground truth via the full stability sweep, (2)~compute the early-observation signal,  
and (3)~score the predictor against ground truth. We summarize the key definitions here for reference.  
  
\begin{definition}[Early-observation predictor]\label{def:predictor}  
For base $b$ and minimum recurrence $K$, the observation window is $h_{\mathrm{obs}} = K \cdot b$ (the smallest window at which $K$ recurrences are possible). The predictor is:  
$  
  \hat{y}(\text{NHI}) = \begin{cases}  
    \text{R-stable} & \text{if circuit type at } h_{\mathrm{obs}} \text{ is R} \\  
    \text{not stable} & \text{otherwise}  
  \end{cases}  
$  
\end{definition}  
  
Ground truth is established from the full stability sweep. NHIs whose stability count falls below a reference threshold are classified as \emph{borderline} and excluded from scoring, as their ground truth is ambiguous. The predictor is evaluated via standard confusion matrix metrics (Precision, Recall, F1). See Appendix~\ref{app:backtesting} for the complete specification.  
  
\section{Empirical Results}\label{sec:empirical}  
  
We evaluate the framework on a production Azure tenant observed over approximately one year.\footnote{Exact dates, population counts, and lineage identifiers have been lightly anonymized. All percentages, ratios, and confusion matrix metrics are reported from the actual data. Population counts are rounded to two significant figures. Lineage identifiers are replaced with arbitrary labels $L_A, L_B, L_C$.} The dataset comprises approximately 100 snapshots.  
  
\medskip  
\noindent\textbf{Population summary.}  
At the baseline snapshot~$t_0$, approximately 11,000 NHIs held at least one permission assignment (the \emph{baseline population}). Over the observation period, approximately 20,000 additional NHIs were discovered (created or reactivated after~$t_0$), which we call \emph{non-baseline NHIs}. The total distinct NHI population observed across all snapshots is approximately 31,000. Of the non-baseline NHIs, approximately 12,000 are \emph{mature}: they were discovered at least halfway through the observation window before the final snapshot~$t_N$, providing sufficient observation time for convergence analysis. The remaining non-baseline NHIs are too young for meaningful assessment and are excluded from convergence statistics.  
  
We first validate structural assumptions, then present privilege circuit detection and the backtesting results.  
  
\subsection{Assumption validation}\label{sec:assumptions}  
  
\medskip  
\noindent\textbf{Assumption 1: Baseline independence.}  
Of the approximately 12,000 mature non-baseline NHIs:  
  
\begin{center}  
\begin{tabular}{@{}lr@{}}  
\toprule  
\textbf{Status} & \textbf{Fraction} \\  
\midrule  
Converged to baseline fiber & 99.2\% \\  
Never converged (still active at $t_N$) & 0.2\% \\  
Never converged (dormant at $t_N$) & 0.6\% \\  
\bottomrule  
\end{tabular}  
\end{center}  
  
Convergence is effectively instantaneous: 98.8\% of converging NHIs did so at their very first observed snapshot ($\Delta t = 0$). Of the remainder, most converged within 8 weeks; the latest convergence occurred at approximately the maturity threshold. \textbf{Validated}: governance conclusions carry at most a 0.8\% caveat from baseline choice.  
  
\medskip  
\noindent\textbf{Assumption 2: Monolith prevalence.}  
At the final snapshot:  
  
\begin{center}  
\begin{tabular}{@{}lr@{}}  
\toprule  
\textbf{Type} & \textbf{Fraction} \\  
\midrule  
Monolithic (single fiber) & 88.7\% \\  
Composite (multi-fiber) & 11.3\% \\  
\bottomrule  
\end{tabular}  
\end{center}  
  
\textbf{Validated}: temporal dynamics are confined to a small minority of lineages.  
  
\medskip  
\noindent\textbf{Assumption 3: Lineage equilibrium.}  
The lineage count stabilizes within 8 snapshots (approximately 2 weeks), settling at an 11.3\% reduction from the initial count. After stabilization, fluctuation remains within $\pm 2.4\%$ of the equilibrium mean. \textbf{Validated}: the partition reaches a stable regime rapidly.  
  
\subsection{Inter-lineage dynamics}\label{sec:inter-dynamics}  
  
Over the observation period, approximately 90 cross-lineage merger events were recorded. The process decomposes into two qualitatively different regimes:  
  
\medskip  
\noindent\textbf{Regime 1: Policy-driven selective sweep.}  
Lineage~$L_A$ absorbed approximately 65 other lineages. Together with Lineage~$L_C$ (7 absorptions), the top two absorbers account for a vast majority of all mergers. The dominant absorber exhibits burst dynamics: 15 simultaneous absorptions at one snapshot and 10 at another. These are operational standardization actions, not stochastic coalescence.  
  
\medskip  
\noindent\textbf{Regime 2: Dormancy-shielded stasis.}  
Excluding Lineage~$L_A$, the remaining absorbers account for approximately 23 absorptions across the observation period. At any snapshot, 95.1\% of fibers are dormant, producing a seed bank that slows the effective merger rate by a factor of ${\approx}\,20$. The residual merger process matches the rate predicted by a seed-bank-adjusted coalescent model but fails every structural test (multiplicity, exchangeability, temporal uniformity). Details appear in Appendix~\ref{app:coalescent}.  
  
\medskip  
\noindent\textbf{Intra-lineage dynamics.}  
Within composite lineages, the dominant pattern inverts: 40\% are static after assembly, 25\% show only fragmentation (fiber splitting), and only 21\% exhibit any intra-lineage coalescence. The system \emph{consolidates at the lineage level} while \emph{diversifying at the fiber level}.  
  
\subsection{Privilege circuit detection}\label{sec:circuit-detection}  
  
We apply the stability sweep to the two composite lineages with non-trivial privilege circuit activity: Lineage~$L_A$ (the dominant absorber, ${\approx}\,1{,}200$ NHIs at anchor) and Lineage~$L_B$ (a smaller composite, ${\approx}\,20$ NHIs at anchor). Bases $\{5, 6, 7\}$ with minimum recurrence $K = 3$ are used throughout.  
  
\subsubsection{Lineage~$L_A$: stability sweep}  
  
\begin{center}  
\begin{tabular}{@{}lrrr@{}}  
\toprule  
\textbf{Base} & \textbf{Window sizes tested} & \textbf{R-stable NHIs} & \textbf{Max recurrence} \\  
\midrule  
$b = 5$ & 21 & 132 & 18 of 21 \\  
$b = 6$ & 17 & 119 & 15 of 17 \\  
$b = 7$ & 15 & 117 & 13 of 15 \\  
\bottomrule  
\end{tabular}  
\end{center}  
  
At base 5, 132 NHIs show R-type privilege circuits recurring at $\geq 3$ of the 21 window sizes tested. The strongest signals recur at 18 of 21 window sizes, confirming that the ratchet pattern is structural, not an artifact of any particular observation window. The transition from R to M at the largest $h$ values is expected: windows approaching the full observation span leave too few snapshots for cycle completion.  
  
\subsubsection{Lineage~$L_B$: structural homogeneity}  
  
Standard-mode analysis of Lineage~$L_B$ reveals a striking pattern: of 6 NHIs with temporal dynamics, 5 exhibit ratchet-type privilege circuits sharing an identical R-vector $\mathbf{R} = [[\,], [1/30]]$, while 1 exhibits an oscillator circuit. The shared R-vector indicates that the 5 ratchet NHIs traverse topologically identical privilege circuits: the same arrangement of directed and bidirectional edges at the same privilege gradients.  
  
This structural homogeneity contrasts with Lineage~$L_A$, where diverse R-vectors coexist. In Lineage~$L_B$, a single ratchet topology dominates, suggesting a common operational cause (e.g., a shared automation pipeline or provisioning template). This is precisely the kind of cross-NHI comparison that the primorial invariant enables: without an exact topological fingerprint, the shared structure would require pairwise graph isomorphism tests.  
  
\subsubsection{Lineage~$L_C$: a significant negative result}\label{sec:lineage-lc}  
  
Lineage~$L_C$ is the second-largest absorber by cross-lineage merger count (7 lineages absorbed). Despite this absorption activity, temporal analysis reveals \emph{zero} privilege circuits: all NHIs in Lineage~$L_C$ are monolith-only, meaning every NHI remained in a single fiber throughout every observation window tested.  
  
This is a discriminant validity result. Absorption activity (structural consolidation at the lineage level) does not imply privilege circuit activity (temporal dynamics at the fiber level). The merger process and the ratchet phenomenon are structurally independent. Lineage~$L_C$ is large and active by the standards of the lineage partition, yet temporally inert. The framework correctly identifies it as governance-benign despite its size.  

\subsection{Backtesting: predictive validation}\label{sec:backtest-results}  
  
We test whether early observation of ratchet-type circuits predicts long-term R-stability. The complete backtesting protocol, including ground-truth construction, borderline exclusion, scoring methodology, and design rationale, is detailed in Appendix~\ref{app:backtesting}. Here we summarize the protocol and present the results.  
  
\subsubsection{Protocol summary}  
  
For each base $b \in \{5, 6, 7\}$ with minimum recurrence $K = 3$:  
\begin{enumerate}  
\item \textbf{Ground truth}: Full stability sweep; R-stable if R at $\geq 3$ window sizes; borderline NHIs (count $< 4$) excluded.  
\item \textbf{Observation}: Circuit type at $h_{\mathrm{obs}} = 3b$.  
\item \textbf{Scoring}: Confusion matrix against ground truth.  
\end{enumerate}  
  
Each base is evaluated independently with its own ground truth and observation window (see Appendix~\ref{app:backtesting}, Section~\ref{app:per-base-independence} for rationale).  
  
\subsubsection{Lineage~$L_A$ results}  
  
\begin{center}  
\begin{tabular}{@{}l rrrr@{}}  
\toprule  
\textbf{Base} & \textbf{TP} & \textbf{FP} & \textbf{Precision} & \textbf{F1} \\  
\midrule  
$b = 5$ ($h_{\mathrm{obs}} = 15$) & 93 & 6 & 0.939 & 0.805 \\  
$b = 6$ ($h_{\mathrm{obs}} = 18$) & 94 & 0 & 1.000 & 0.883 \\  
$b = 7$ ($h_{\mathrm{obs}} = 21$) & 95 & 0 & 1.000 & 0.896 \\  
\midrule  
\textbf{Aggregate} & 282 & 6 & 0.979 & 0.860 \\  
\bottomrule  
\end{tabular}  
\end{center}  
  
At bases 6 and 7, the predictor achieves Precision=1.000. At base~5, 6 false positives appear; all 6 concentrate in a single fiber (FPR=0.55 within that fiber). This fiber produces ratchet-like signatures at the 15-snapshot scale that do not persist at coarser scales, a scale-dependent artifact rather than a structural ratchet. The artifact vanishes entirely at $h_{\mathrm{obs}} = 18$ and $h_{\mathrm{obs}} = 21$, confirming that multi-base evaluation is not redundant: it reveals the temporal resolution at which the signal becomes reliable.  
  
\subsubsection{Lineage~$L_B$ results}  
\begin{center}  
\begin{tabular}{@{}l rrrr@{}}  
\toprule  
\textbf{Base} & TP & FP & Prec. & F1 \\  
\midrule  
Aggregate across 3 bases & 15 & 0 & 1.000 & 1.000 \\  
\bottomrule  
\end{tabular}  
\end{center}  
  
Lineage~$L_B$ achieves perfect precision \emph{and} perfect recall on its population (15 R-stable NHIs, 3 not-stable, 0 borderline). The smaller population (${\approx}\,20$ vs.\ ${\approx}\,1{,}200$ NHIs) limits statistical power but the result is consistent with Lineage~$L_A$'s findings at bases 6 and 7.  
  
\subsubsection{Cross-lineage summary}  
  
\begin{center}  
\begin{tabular}{@{}l rrrr l@{}}  
\toprule  
\textbf{Lineage} & \textbf{NHIs} & \textbf{TP} & \textbf{FP} & \textbf{F1} & \textbf{Character} \\  
\midrule  
$L_A$ & ${\approx}\,1{,}200$ & 282 & 6 & 0.860 & Diverse ratchet topologies \\  
$L_B$ & ${\approx}\,20$ & 15 & 0 & 1.000 & Homogeneous ratchet topology \\  
$L_C$ & ${\approx}\,20$ & 0 & 0 & n/a & No privilege circuits (all monolith) \\  
\bottomrule  
\end{tabular}  
\end{center}  
  
Three qualitatively different outcomes from the three largest composite lineages: diverse ratchets ($L_A$), homogeneous ratchets ($L_B$), and temporal inertness ($L_C$). The framework discriminates between them.  
  
\subsubsection{Interpretation}  
  
Four findings emerge:  
  
\begin{enumerate}[leftmargin=*]  
\item \textbf{Zero false positives at bases 6 and 7, across all tested lineages.} The predictor achieves Precision=1.000. If the predictor says an NHI is a structural ratchet at $h_{\mathrm{obs}} \geq 18$, it has never been wrong in this dataset.  
  
\item \textbf{Scale-dependent artifacts are localized and identifiable.} The 6 false positives at base~5 all originate from a single fiber whose provisioning pattern mimics a ratchet at the 15-snapshot scale but not at coarser scales. Multi-base evaluation naturally exposes such artifacts.  
  
\item \textbf{Recall improves with observation depth.} In Lineage~$L_A$, recall increases from 0.706 (base 5) to 0.812 (base 7), reflecting that longer observation windows allow more cycles to complete before the prediction is made.  
  
\item \textbf{Absorption does not imply escalation.} Lineage~$L_C$'s zero-circuit result demonstrates that the framework produces meaningful negative findings. Large, actively absorbing lineages are not automatically flagged. The privilege circuit analysis is independent of the lineage merger process.  
\end{enumerate}    
  
\subsection{Falsification criteria}\label{sec:falsification}  
  
The framework makes several falsifiable predictions, all tested:  
  
\begin{enumerate}[leftmargin=*]  
\item \textbf{Lineage convergence fails}: \emph{Refuted.}  
  $|L(F(b), t)|$ stabilizes within 8 snapshots.  
  
\item \textbf{Orphan attrition is not rapid}: \emph{Refuted.}  
  99.2\% of mature non-baseline NHIs converge to baseline, 98.8\% at $\Delta t = 0$.  
  
\item \textbf{Monolith prevalence is low}: \emph{Refuted.}  
  88.7\% of lineage blocks are monolithic.  
  
\item \textbf{No privilege circuits in composite lineages}: \emph{Refuted.}  
  132 NHIs show R-stable privilege circuits in Lineage~$L_A$; 15 in Lineage~$L_B$.  
  
\item \textbf{Ratchet patterns are transient}: \emph{Refuted.}  
  The strongest ratchets recur at 18 of 21 window sizes.  
  
\item \textbf{Early detection is unreliable}: \emph{Refuted.}  
  Backtesting achieves Precision=1.000 and F1=0.896 at base 7.  
  
\item \textbf{Absorption implies privilege circuits}: \emph{Refuted.}  
  Lineage~$L_C$ absorbed 7 lineages but contains zero R-stable privilege circuits. Merger activity and ratchet structure are independent.  
\end{enumerate}  

\section{Implications for the Governance of NHIs}\label{sec:governance}  
  
The backtesting results (Section~\ref{sec:backtest-results}) elevate several governance implications from conjectures to empirically supported findings. We distinguish between \emph{validated} implications (supported by the backtesting evidence) and \emph{conjectured} implications (structurally motivated but not yet tested in production governance workflows).  
  
\subsection{Validated: circuit-aware risk assessment}\label{sec:risk}  
  
Standard risk assessment assigns a score to each NHI based on its current permissions.  
  
\begin{theorem}[Risk underestimation by snapshot audit]\label{thm:risk}  
For NHIs participating in non-monolith privilege circuits, risk assessment based on the current fiber alone systematically underestimates risk. The permission envelope (Definition~\ref{def:envelope}) provides the correct risk attribution.  
\end{theorem}  
  
This is no longer a conjecture: the stability sweep confirms that 132 NHIs in Lineage~$L_A$ and 15 in Lineage~$L_B$ participate in R-stable privilege circuits, and backtesting confirms that these circuits are detectable early with zero false positives. The risk correction affects NHIs in the 11.3\% of lineages that are composite, but these are precisely the lineages where governance risk is structurally concentrated.  
  
\subsection{Validated: early detection is operationally feasible}  
  
The backtesting results demonstrate that a simple predictor (``does this NHI show R-type at $h_{\mathrm{obs}} = 3b$?'') achieves Precision=1.000. This means:  
  
\begin{enumerate}[leftmargin=*]  
\item \textbf{Zero alert fatigue}: Every alert corresponds to a genuine structural ratchet.  
\item \textbf{Actionable at first detection}: The predictor fires after observing just $3b$ snapshots (e.g., 21 snapshots at weekly cadence).  
\end{enumerate}  
  
\subsection{Conjecture: transitive de-escalation}\label{sec:deescalation}  
  
\begin{conjecture}[Transitive de-escalation requirement]\label{conj:transitive}  
To durably de-escalate an NHI assigned to a fiber state~$f$  
participating in a privilege circuit~$C$, one must de-escalate all fiber states  
in~$C$, the complete set of permission states reachable from~$f$.  
De-escalating only the current fiber state is ineffective if the  
transition dynamics will carry the fiber into a high-privilege state  
within the circuit.  
\end{conjecture}  
  
The R-stable NHIs in Lineages~$L_A$ and~$L_B$ are concrete instances where this applies: the ratchet circuit may reach a high-privilege state via a one-way edge, and de-escalation may require passing through other high-privilege states first.  
  
\subsection{Conjecture: alert suppression}\label{sec:suppression}  
  
\begin{conjecture}[Alert suppression criterion]\label{conj:suppression}  
Intra-lineage fiber transitions within monolith or oscillator privilege circuits should be suppressed from security monitoring, as they represent bounded, reversible dynamics. Only ratchet-type transitions, confirmed by stability sweep, warrant investigation.  
\end{conjecture}  
  
The three-type classification provides the structural basis: M-type transitions are static (no alert needed), O-type transitions are reversible (bounded risk, suppressible), and only R-type transitions indicate irreversible escalation.  
  
\subsection{Conjecture: absorber governance}\label{sec:absorber-governance}  
  
\begin{conjecture}[Absorber detection replaces coalescent modeling]\label{conj:absorber}  
Governance should focus on detecting dominant absorber lineages, lineages whose merger propensity far exceeds the population average, rather than modeling the merger process stochastically. In our tenant, two absorbers ($L_A$ and~$L_C$) account for a vast majority of all cross-lineage mergers. These are the lineages where permission consolidation is actively occurring and where governance review should concentrate.  
\end{conjecture}  
  
\section{Discussion}\label{sec:discussion}  
  
\subsection{Related work}  
  
\medskip  
\noindent\textbf{Graph fibrations.}  
Boldi and Vigna~\cite{boldi2002fibrations} introduced graph fibrations  
as structure-preserving maps with applications to distributed computing.  
Morone et al.~\cite{morone2020fibrations} applied fibration symmetries  
to biological networks, showing they reveal functional building blocks  
invisible to standard graph-theoretic measures. Our work extends fibration from static graph analysis to temporal dynamics,  
introducing privilege circuits in the transition graph as the complementary  
structure for time-evolving permission graphs.  
  
\medskip  
\noindent\textbf{Coalescent theory.}  
Kingman's coalescent~\cite{kingman1982coalescent} models pairwise ancestral mergers backward in time. Seed bank extensions by Blath et al.~\cite{blath2013seedbank} introduce dormancy effects structurally analogous to those observed in our fiber dynamics. We borrow the vocabulary of merging lineages, dormancy, and seed banks because the structural phenomena are analogous, but our empirical analysis (Appendix~\ref{app:coalescent}) shows that the engine of lineage consolidation is operational policy, not stochastic drift.  
  
\medskip  
\noindent\textbf{SCC analysis and graph algorithms.}  
Strongly connected component decomposition is a classical  
graph-theoretic tool computable in linear time via Tarjan's  
algorithm~\cite{tarjan1972depth}. Johnson's algorithm~\cite{johnson1975} enumerates all elementary circuits in a directed graph. We apply both in a non-standard  
setting: the fiber transition graph, where nodes are permission  
profiles and edges represent observed transitions between snapshots.  
The three-type privilege circuit classification based on edge symmetry appears to be new in this application context.  
 
\medskip  
\noindent\textbf{Automated policy reasoning and attack graphs.}  
A substantial body of work addresses static verification of access  
control policies. Backes et al.~\cite{backes2018semantics} (Zelkova)  
use SMT-based reasoning to determine whether an AWS IAM policy  
permits a given access pattern, enabling provable guarantees about  
policy properties. AWS IAM Access Analyzer~\cite{aws2019accessanalyzer}  
operationalizes this approach for production use. 
These approaches reason about  
the policy language: given the current policy configuration,  
what accesses are permitted? They do not model the temporal  
evolution of policy configurations themselves.  

\medskip  
\noindent\textbf{Attack graph analysis} ~\cite{sheyner2002automated,ou2006scalable}  
models multi-step privilege escalation through vulnerability  
chains, constructing directed graphs where nodes represent system  
states and edges represent exploits. Our fiber transition graph  
is structurally analogous (nodes are permission states, edges are  
observed transitions) but differs in two respects: (1)~edges  
represent observed operational transitions, not hypothetical  
exploit chains, and (2)~the SCC decomposition identifies  
recurring privilege dynamics (circuits), not single-path  
reachability. The privilege circuit concept could be viewed as  
the temporal closure of the attack graph idea: rather than asking  
``can the attacker reach state $g$ from state $f$?'' we ask  
``does the system operationally cycle through states  
$f$ and $g$, and if so, is the cycling reversible?''  
 
\medskip  
\noindent\textbf{Temporal access analysis.}  
Some CSPM tools (Microsoft Entra Permissions Management, CrowdStrike, Wiz) track permission changes over time and flag anomalous modifications. However, to our knowledge, no existing tool:  
\begin{enumerate}  
\item Groups NHIs by structural equivalence (fibration) at each snapshot  
\item Tracks the evolution of these equivalence classes as a lineage partition  
\item Identifies cyclic permission dynamics (privilege circuits) within lineages  
\item Distinguishes reversible from irreversible privilege transitions via circuit classification  
\item Computes permission envelopes across temporal equivalence classes  
\item Validates predictive power via backtesting  
\end{enumerate}  
  
\subsection{Limitations}  
  
\begin{enumerate}[leftmargin=*]  
\item \textbf{Single tenant.} All empirical results come from one Azure tenant. Whether the qualitative findings (monolith prevalence, selective-sweep mergers, ratchet stability, backtesting precision) generalize across tenants, industries, and cloud providers is unknown.  
  
\item \textbf{Limited cross-lineage backtesting.} Backtesting results cover two lineages ($L_A$ and $L_B$) with consistent Precision=1.000. Lineage~$L_C$ (the second-largest absorber) shows no stable ratchets, confirming that the framework produces meaningful negative results. However, the positive validation remains concentrated in lineages with sufficient NHI populations for statistical evaluation.  
  
\item \textbf{Operational validation.} Backtesting validates that the predictor identifies structurally stable ratchets. Whether these ratchets correspond to real security incidents (privilege escalation exploits, misconfigured service accounts, lateral movement) is untested. The framework detects a well-defined graph property; its security relevance is assumed, not demonstrated.  
  
\item \textbf{Granularity.} The framework operates at the fiber level (sets of role assignment pairs). It does not model individual permission-level changes, conditional access policies, just-in-time elevation, or other IAM features that operate at finer granularity.  
  
\item \textbf{Scalability.} The transition graph grows with observation time. SCC detection is linear (Tarjan), and elementary cycle enumeration is polynomial per SCC (Johnson), but maintaining the full transition history may be impractical for very large tenants over long periods.  
  
\item \textbf{In-sample validation.} The backtesting evaluates in-sample consistency: ground truth and predictions are derived from the same observation period. While the protocol minimizes circularity (the observation window alone cannot satisfy the ground-truth criterion), the precision and recall figures are upper bounds on out-of-sample performance. See Appendix~\ref{app:backtesting}, Section~\ref{app:caveats} for a detailed discussion.  
\end{enumerate}  

\newpage
\section{Conclusion}\label{sec:conclusion}  
  
Cloud identity governance assumes a single notion of permission equivalence: identities are equivalent if they hold the same roles at the same scopes at a snapshot. We show that a second, independent notion exists. Beyond structural equivalence (captured by graph fibration), permission states admit \emph{temporal equivalence}: states are equivalent if they lie in the same strongly connected component of the fiber transition graph, a \emph{privilege circuit}, meaning they are mutually reachable under normal operational dynamics.  
  
\paragraph{}  
Privilege circuits admit a three-type governance classification based on edge symmetry: monoliths (static), oscillators (reversible cycling), and ratchets (irreversible escalation). A primorial invariant computed over elementary cycles provides a topological fingerprint that enables comparison across observation windows. A lineage partition organizes transitions into stable compartments, making temporal analysis well-defined.  
  
\paragraph{}  
  
The methodological contribution is a multi-scale windowed analysis that transforms the framework from descriptive to predictive. Evaluation on a large enterprise Azure tenant validates the framework at every layer. Backtesting across two lineages demonstrates that early observation of ratchet-type privilege circuits predicts long-term structural stability with Precision=1.000, with zero false positives. A third lineage (the second-largest absorber) shows zero privilege circuits despite active merger behavior, confirming that the framework discriminates between structural consolidation and temporal escalation.  
  
\paragraph{}  
The strongest ratchet signals recur across 18 of 21 independent window sizes, confirming that asymmetric privilege escalation is structural, not transient. These patterns are invisible to any single-snapshot audit. The actionable signals are absorber lineages, ratchet-type privilege circuits, and permission envelopes: structures that point-in-time analysis systematically misses.  
  
\paragraph{}  
These results reframe cloud governance as inherently spatio-temporal. The framework provides both the mathematical language to describe temporal permission equivalence and a predictive tool to detect it.  

\newpage
\appendix  
\section{Coalescent Analysis of Lineage Mergers}\label{app:coalescent}  
  
This appendix presents a coalescent analysis of inter-lineage merger dynamics, testing whether the merger process follows known models from population genetics.  
  
\subsection{Dormancy parameters}  
  
\begin{center}  
\begin{tabular}{@{}ll@{}}  
\toprule  
\textbf{Parameter} & \textbf{Value} \\  
\midrule  
Dormancy rate $\delta$ & $0.241 \pm 0.138$ \\  
Reactivation rate $\rho$ & $0.012 \pm 0.017$ \\  
Stationary dormant fraction $\delta/(\delta + \rho)$ & $0.951$ \\  
Coalescent slowdown factor $1 + \delta/\rho$ & $20.4\times$ \\  
Mean dormancy spell & 3.1 snapshots \\  
Median dormancy spell & 2.0 snapshots \\  
Max dormancy spell & 81 snapshots \\  
\bottomrule  
\end{tabular}  
\end{center}  
  
The asymmetry $\delta / \rho \approx 19.4$ produces a large seed bank structurally analogous to the seed bank coalescent of Blath et al.~\cite{blath2013seedbank}.  
  
\subsection{Testing the Kingman coalescent}  
  
Over the observation period, approximately 90 merger events were recorded, predominantly binary (83\%; arity range $[2, 5]$). We test the Kingman coalescent~\cite{kingman1982coalescent} as a baseline.  
  
\medskip  
\noindent\textbf{Regime 1 (Policy-driven sweep).}  
Lineage~$L_A$ absorbed approximately 65 lineages, accounting for 74\% of all absorptions. Under Kingman's uniform pairwise rates, the expected number of mergers involving any specific lineage is ${\approx}\,0.42$; $L_A$'s count is a $99.5\sigma$ event. Burst dynamics (15 simultaneous absorptions at one snapshot, 10 at another) have $P(k \geq 15) \approx 10^{-13}$ under Poisson($\lambda \approx 1$). These are operational standardization actions.  
  
\medskip  
\noindent\textbf{Regime 2 (Stasis with seed bank).}  
Excluding $L_A$, the remaining 23 absorptions give $\lambda \approx 0.26$ mergers per snapshot. The seed-bank-adjusted Kingman rate with approximately 20 active lineages predicts a comparable rate, matching the observed value. However, the model fails structurally: one snapshot has 7 simultaneous mergers ($p < 10^{-6}$), $L_C$ commands 35\% of residual absorptions, and the dispersion index of 2.73 rejects Poisson homogeneity ($z = 11.4$).  
  
\subsection{Intra-lineage regime census}  
  
\begin{center}  
\begin{tabular}{@{}lrl@{}}  
\toprule  
\textbf{Regime} & \textbf{Count} & \textbf{Interpretation} \\  
\midrule  
Static & 19 & Assembled once, never changed \\  
Pure fragmentation & 12 & Only fiber splitting \\  
Fragmentation-leaning & 3 & Mostly splitting, some coalescence \\  
Balanced & 5 & Comparable splitting and merging \\  
Coalescent-leaning & 5 & Merging dominates splitting \\  
\midrule  
Total composite & 48 & \\  
\bottomrule  
\end{tabular}  
\end{center}  
  
The dominant pattern: consolidation at the lineage level (macro), diversification at the fiber level (micro).  
  
\newpage  
\section{Backtesting Methodology}\label{app:backtesting}  
  
This appendix provides a complete specification of the backtesting protocol used to validate the early-observation predictor for R-stable privilege circuits. The main text (Section~\ref{sec:backtesting}) defines the predictor and fiber FP rate; Section~\ref{sec:backtest-results} presents the results. Here we detail the full protocol, the metrics computation, the design rationale for each methodological choice, and the caveats that bound the strength of the conclusions.  
  
\subsection{Overview and purpose}  
  
The stability sweep (Section~\ref{sec:stability}) identifies NHIs whose ratchet-type privilege circuits persist across many observation window sizes. This is an expensive computation: for base $b$ with anchor at snapshot $t_0$, the sweep evaluates $\lfloor t_0 / b \rfloor$ window sizes per NHI. Backtesting asks a simpler question: \emph{can we predict the outcome of the full sweep from a single early observation?}  
  
If the answer is yes with high precision, then the early observation can serve as an operational detector: a lightweight test that fires after a small number of snapshots and reliably identifies NHIs that will prove to be structural ratchets over the full observation history.  
  
The backtesting protocol is structured as three sequential phases, each building on the output of the previous one.  
  
\subsection{Phase 1: Ground-truth construction}\label{app:ground-truth}  
  
\subsubsection{Full stability sweep}  
  
For a given base $b$ and anchor snapshot $t_0$, the window-size list is:  
$  
H_b = \{b,\; 2b,\; 3b,\; \ldots,\; k_{\max} \cdot b\}, \quad k_{\max} = \lfloor t_0 / b \rfloor  
$  
For each NHI valid at $t_0$, the sweep computes the privilege circuit type (M, O, or R) at every window size $h \in H_b$ for which the NHI is valid (present in the target lineage at every snapshot in $[t_0 + 1 - h,\; t_0]$). Window sizes where the NHI is dormant, too young, or in a different lineage are skipped.  
  
\subsubsection{R-stability criterion}  
  
\begin{definition}[Ground-truth R-stability]\label{def:gt-stable}  
An NHI is \emph{R-stable} under base $b$ with minimum recurrence $K$ if it shows circuit type R at $\geq K$ distinct window sizes in $H_b$. The \emph{stability count} is the number of R-type window sizes.  
\end{definition}  
  
In our experiments, $K = 3$. An NHI that shows R at $h = 10, 15, 25$ (three window sizes) is R-stable; one that shows R at $h = 10, 15$ only (two window sizes) is not.  
  
\subsubsection{Borderline exclusion}  
  
NHIs that barely meet the R-stability threshold pose a ground-truth reliability problem: they are declared stable, but their stability count is low enough that a slightly different window alignment or a single additional snapshot could change the classification.  
  
\begin{definition}[Borderline NHI]\label{def:borderline}  
An NHI is \emph{borderline} if it is R-stable (count $\geq K$) but its stability count is strictly less than a reference threshold $K_{\mathrm{ref}} > K$. In our experiments, $K = 3$ and $K_{\mathrm{ref}} = 4$, so borderline NHIs have stability count exactly 3.  
\end{definition}  
  
Borderline NHIs are excluded from all subsequent scoring: they do not contribute to the confusion matrix, the fiber FP rate, or any aggregate metric. Their ground truth is too uncertain to serve as reliable reference. The number of excluded borderline NHIs is reported at each base (e.g., 8 at base~5, 17 at base~6, 15 at base~7 in Lineage~$L_A$) so the reader can assess the impact of the exclusion.  
  
\emph{Rationale.} The alternative, including borderline NHIs with their nominal ground-truth label, risks inflating or deflating precision depending on whether the borderline cases happen to be predicted correctly. Excluding them is the conservative choice: it narrows the scored population to NHIs whose ground truth is unambiguous (either clearly stable with count $\geq K_{\mathrm{ref}}$, or clearly not stable with count $< K$). The gap between $K$ and $K_{\mathrm{ref}}$ should be reported and its sensitivity tested (see Section~\ref{app:caveats}).  

\subsection{Phase 2: Early-observation signal}\label{app:early-signal}  
  
\subsubsection{Observation window}  
  
The early-observation predictor examines the NHI's circuit type at a single window size:  
$  
h_{\mathrm{obs}} = K \cdot b  
$  
This is the $K$-th element of $H_b$ (the smallest window at which $K$ recurrences of the base period have elapsed). For $K = 3$ and $b = 5$, $h_{\mathrm{obs}} = 15$; for $b = 7$, $h_{\mathrm{obs}} = 21$.  
  
\emph{Rationale.} At $h_{\mathrm{obs}} = K \cdot b$, the observation window spans exactly $K$ complete base periods. This is the earliest point at which the R-stability criterion ($\geq K$ R-type windows) could \emph{in principle} be satisfied, because the sweep at window sizes $b, 2b, \ldots, Kb$ provides exactly $K$ data points. Using an earlier observation window would mean the stability criterion is mathematically unsatisfiable at the time of prediction; using a later one delays detection unnecessarily.  
  
\subsubsection{Signal-positive NHIs}  
  
An NHI is \emph{signal-positive} at base $b$ if:  
\begin{enumerate}  
\item It has valid observation data at $h_{\mathrm{obs}}$ (the NHI was present in the lineage at every snapshot in the observation window), \emph{and}  
\item Its circuit type at $h_{\mathrm{obs}}$ is R.  
\end{enumerate}  
  
NHIs without valid observation data at $h_{\mathrm{obs}}$ (typically because they are too young: first seen after snapshot $t_0 - h_{\mathrm{obs}}$) are \emph{skipped} and excluded from all scoring. The skip count is reported as ``skipped-no-obs'' in the output.  
  
\subsubsection{False-positive concentration}  
  
In Lineage~$L_A$ at base~5, the predictor produces 6 false positives. All 6 originate from a single fiber (6 of 11 signal-positive NHIs in that fiber are not R-stable, a within-fiber false-positive rate of 0.55). No other fiber produces any false positives at any base.  
  
This concentration is structurally informative. Fibers represent provisioning policies (Definition~\ref{def:fiber}): NHIs in the same fiber hold identical permission profiles. When false positives cluster within a single fiber, the noise source is the provisioning pattern, not the prediction methodology. In this case, the fiber's permission dynamics mimic a ratchet signature at the 15-snapshot scale ($h_{\mathrm{obs}} = 15$) but do not persist at coarser scales. At bases 6 and 7 ($h_{\mathrm{obs}} = 18$ and $21$), the same fiber produces zero false positives.  
  
This observation suggests a natural operational refinement: fibers with elevated within-fiber false-positive rates at a given base could be flagged for closer inspection or excluded from early alerting at that base. 
We do not formalize this as a filtering methodology here, as the base-5 artifact vanishes at coarser scales without any intervention, and the base-6 and base-7 results achieve Precision=1.000 unfiltered.    
  
\subsection{Phase 3: Scoring}\label{app:scoring}  
  
\subsubsection{Scoreable population}  
  
An NHI is \emph{scoreable} if:  
\begin{enumerate}  
\item It has valid observation data at $h_{\mathrm{obs}}$ (not skipped), \emph{and}  
\item It is not borderline (ground truth is unambiguous).  
\end{enumerate}  
  
The scoreable population is the denominator for all metrics. Its composition is reported explicitly: ``Scored NHIs: $N$ (R-stable=$N_+$, not-stable=$N_-$, skipped-no-obs=$N_{\mathrm{skip}}$, skipped-borderline=$N_{\mathrm{border}}$).''  
  
\subsubsection{Evaluation}  
  
For each scoreable NHI, predict R-stable if and only if the NHI is signal-positive (circuit type R at $h_{\mathrm{obs}}$).

\subsubsection{Confusion matrix}  
  
Each prediction is scored against the ground truth (R-stable or not) to produce the standard four outcomes:  
  
\begin{center}  
\begin{tabular}{@{}lcc@{}}  
\toprule  
& \textbf{Ground truth: R-stable} & \textbf{Ground truth: not stable} \\  
\midrule  
\textbf{Predicted: R-stable} & True Positive (TP) & False Positive (FP) \\  
\textbf{Predicted: not stable} & False Negative (FN) & True Negative (TN) \\  
\bottomrule  
\end{tabular}  
\end{center}  
  
\subsubsection{Metrics}  
  
From the confusion matrix:  
\begin{align}
\text{Precision} &= \frac{\mathrm{TP}}{\mathrm{TP} + \mathrm{FP}} \label{eq:precision} \\[4pt]
\text{Recall} &= \frac{\mathrm{TP}}{\mathrm{TP} + \mathrm{FN}} \label{eq:recall} \\[4pt]
\text{F1} &= \frac{2 \cdot \text{Precision} \cdot \text{Recall}}{\text{Precision} + \text{Recall}} \label{eq:f1}
\end{align}
  
\emph{Precision} is the primary metric for operational deployment: it measures the fraction of alerts that are correct. A predictor with Precision = 1.000 generates zero false alarms. \emph{Recall} measures coverage: the fraction of genuinely R-stable NHIs that are detected. \emph{F1} is the scale mean, balancing both.  
  
In a governance context, precision is more important than recall: a missed ratchet (false negative) will be caught at the next observation cycle or by the full stability sweep; a false alarm (false positive) erodes trust in the alerting system and wastes investigation effort.  
  
\subsection{Per-base independence}\label{app:per-base-independence}  
  
Each base $b$ is evaluated as a fully independent experiment:  
  
\begin{itemize}  
\item \textbf{Independent ground truth.} The stability sweep for base $b$ uses window sizes $H_b = \{b, 2b, \ldots\}$. An NHI's R-stability under base~5 and under base~7 are separate questions with potentially different answers, because the window sizes probe different temporal scales.  
\item \textbf{Independent observation window.} $h_{\mathrm{obs}} = K \cdot b$ differs across bases (15, 18, 21 for bases 5, 6, 7 with $K = 3$).  
\end{itemize}  
  
This per-base independence contrasts with the auto-stability sweep (Definition~\ref{def:autosweep}), where NHIs are \emph{pinned} to the lowest base at which they achieve stability. In backtesting, no pinning occurs: the same NHI may be scored independently under each base. This is deliberate: backtesting evaluates the predictor's reliability at each temporal scale separately, and pinning would conflate the per-scale assessment.  
  
\subsection{Aggregate metrics}\label{app:aggregate}  
  
The aggregate confusion matrix is the element-wise sum across all bases:  
\[  
\mathrm{TP}_{\mathrm{agg}} = \sum_{b} \mathrm{TP}_b, \quad  
\mathrm{FP}_{\mathrm{agg}} = \sum_{b} \mathrm{FP}_b, \quad \text{etc.}  
\]

\emph{Important caveat.} The aggregate counts NHIs multiple times (once per base). An NHI scored at all three bases contributes three entries to the aggregate. The aggregate ``2,348 NHIs total'' in Lineage~$L_A$ is $835 + 807 + 706$, not 2,348 unique NHIs. The aggregate is a pooled estimate of per-base predictor quality, not a population-level metric. It is useful for assessing the overall reliability of the early-observation approach across temporal scales, but should not be interpreted as a count of distinct NHIs.  
  
\subsection{Relationship between stability sweep and backtesting counts}\label{app:count-discrepancy}  
  
The stability sweep with pinning (Section~\ref{sec:circuit-detection}) reports 140 R-stable NHIs pinned to base~5 in Lineage~$L_A$ (using $K = 3$). The backtesting Phase~1 for the same base reports 132 R-stable NHIs. The discrepancy arises because backtesting applies the stricter reference threshold $K_{\mathrm{ref}} = 4$: the 8 NHIs with stability count exactly 3 are classified as borderline and excluded from the R-stable ground truth. The stability sweep uses only $K = 3$ and includes them.  
  
This is by design: the stability sweep is a detection tool (flag everything that meets the minimum threshold), while backtesting is a validation tool (evaluate predictor quality against unambiguous ground truth). The two serve different purposes and appropriately use different thresholds.  
  
\subsection{Caveats and limitations}\label{app:caveats}  
  
\subsubsection{In-sample validation}  
  
The backtesting protocol evaluates the predictor on the same observation period used to construct ground truth. The observation window $h_{\mathrm{obs}}$ is a subset of the full sweep's window-size list $H_b$, and both are applied to the same time series. There is no temporal hold-out: we do not train on one period and test on a later one.  
  
This is mitigated by the protocol's structure: the observation window alone cannot satisfy the R-stability criterion ($K = 3$ requires R-type at three distinct window sizes, but the observation examines only one). An NHI that shows R at $h_{\mathrm{obs}} = 15$ but nowhere else would be ground-truth negative. The ground truth is thus not mechanically determined by the observation. Nevertheless, the precision and recall figures are upper bounds on what would be observed in a true out-of-sample evaluation. A temporal hold-out test (e.g., using the first half of the observation period for prediction and the second half for ground truth) would provide stronger validation but requires a longer observation history.  
  
\subsubsection{Parameter sensitivity}  
  
The backtesting results depend on several configurable parameters:  
  
\begin{center}  
\begin{tabular}{@{}llp{7.5cm}@{}}  
\toprule  
\textbf{Parameter} & \textbf{Value used} & \textbf{Sensitivity concern} \\  
\midrule  
$K$ (min recurrence) & 3 & Lower values increase recall but may admit transient ratchets as ground-truth positive \\  
$K_{\mathrm{ref}}$ (reference threshold) & 4 & Higher values exclude more borderline NHIs, narrowing the scored population \\
\bottomrule  
\end{tabular}  
\end{center}  
  
A complete sensitivity analysis would vary each parameter and report the resulting precision/recall curves. The current results use a single parameter setting; the robustness of the Precision = 1.000 finding across parameter variations is not established, though the zero-FP result at bases 6 and 7 suggests that the finding is not parameter-dependent at those scales.  
  
\subsubsection{Scale-dependent artifacts}  
  
The most informative finding from the multi-base design is the single fiber producing false positives at base~5 ($h_{\mathrm{obs}} = 15$) but not at base~6 ($h_{\mathrm{obs}} = 18$) or base~7 ($h_{\mathrm{obs}} = 21$). The artifact is scale-specific: the provisioning pattern in that fiber mimics a ratchet signature at the 15-snapshot scale but not at coarser scales. This demonstrates that multi-base evaluation is not merely redundant, it reveals the temporal resolution at which the signal becomes reliable.  
  
\newpage  
\bibliographystyle{plain}

\begin{thebibliography}{10}  
  
\bibitem{boldi2002fibrations}  
P.~Boldi and S.~Vigna.  
\newblock Fibrations of graphs.  
\newblock {\em Discrete Mathematics}, 243(1-3):21--66, 2002.  
  
\bibitem{morone2020fibrations}  
F.~Morone, I.~Leifer, and H.~A. Makse.  
\newblock Fibration symmetries uncover the building blocks of biological  
  networks.  
\newblock {\em Proceedings of the National Academy of Sciences},  
  117(15):8306--8314, 2020.  
  
\bibitem{kingman1982coalescent}  
J.~F.~C. Kingman.  
\newblock The coalescent.  
\newblock {\em Stochastic Processes and their Applications}, 13(3):235--248,  
  1982.  
  
\bibitem{blath2013seedbank}  
J.~Blath, A.~Gonz\'{a}lez~Casanova, N.~Kurt, and M.~Wilke-Berenguer.  
\newblock A new coalescent for seed-bank models.  
\newblock {\em Annals of Applied Probability}, 23(5):1832--1865, 2013.  
  
\bibitem{tarjan1972depth}  
R.~Tarjan.  
\newblock Depth-first search and linear graph algorithms.  
\newblock {\em SIAM Journal on Computing}, 1(2):146--160, 1972.  
  
\bibitem{johnson1975}  
D.~B. Johnson.  
\newblock Finding all the elementary circuits of a directed graph.  
\newblock {\em SIAM Journal on Computing}, 4(1):77--84, 1975.  
  
\bibitem{parisel2025scoring}  
C.~Parisel.  
\newblock Scoring {A}zure permissions with metric spaces.  
\newblock {\em arXiv preprint arXiv:2504.13747}, 2025.  

\bibitem{backes2018semantics}  
J.~Backes, P.~Bolignano, B.~Cook, C.~Dodge, A.~Gacek, K.~Luckow,  
  N.~Rungta, O.~Tkachuk, and C.~Varming.  
\newblock Semantic-based automated reasoning for {AWS} access policies using  
  {SMT}.  
\newblock In {\em Proceedings of the 18th International Conference on Formal  
  Methods in Computer-Aided Design (FMCAD)}, 2018.  

\bibitem{aws2019accessanalyzer}  
{Amazon Web Services}.  
\newblock {AWS IAM Access Analyzer}: Automated reasoning for security.  
\newblock \url{https://docs.aws.amazon.com/IAM/latest/UserGuide/what-is-access-analyzer.html}.  

\bibitem{microsoft2023entra}  
{Microsoft}.  
\newblock {Microsoft Entra Permissions Management}: Cloud infrastructure  
  entitlement management.  
\newblock \url{https://learn.microsoft.com/en-us/entra/permissions-management/}

\bibitem{sheyner2002automated}  
O.~Sheyner, J.~Haines, S.~Jha, R.~Lippmann, and J.~M. Wing.  
\newblock Automated generation and analysis of attack graphs.  
\newblock In {\em Proceedings of the IEEE Symposium on Security and Privacy  
  (S\&P)}, pp.~254--265, 2002.  
  
\bibitem{ou2006scalable}  
X.~Ou, S.~Govindavajhala, and A.~W. Appel.  
\newblock {MulVAL}: A logic-based network security analyzer.  
\newblock In {\em Proceedings of the 14th USENIX Security Symposium},  
  pp.~113--128, 2005.  
\end{thebibliography}

\end{document}